\documentclass[submission,copyright,creativecommons]{eptcs}
\providecommand{\event}{SOS 2007} 
\usepackage{breakurl}             
\usepackage{underscore}           

\usepackage{amsthm,amsmath,amssymb}

\usepackage{url}
\usepackage{graphicx}
\usepackage{etoolbox}
\usepackage{inputenc}
\usepackage[T1]{fontenc}
\usepackage{bussproofs}
 
\usepackage{commandsEmiliano}
\usepackage[linesnumbered,ruled]{algorithm2e}
\DontPrintSemicolon
\AtBeginEnvironment{align}{\setcounter{equation}{0}}

\newtheorem{theorem}{Theorem}

\newtheorem{lemma}{Lemma}

\newtheorem{proposition}[theorem]{Proposition}
\newtheorem{definition}[theorem]{Definition}

\title{Exploiting Belief Bases \\for Building Rich Epistemic Structures}
\author{Emiliano Lorini
\institute{IRIT, CNRS, Toulouse University\\ Toulouse, France}
\email{Emiliano.Lorini@irit.fr}
 }
\def\titlerunning{Exploiting Belief Bases for Building Rich Epistemic Structures}
\def\authorrunning{E. Lorini}
\begin{document}
\maketitle

\begin{abstract}
We introduce a semantics for epistemic logic exploiting a belief base abstraction. 
Differently from existing Kripke-style semantics for epistemic logic in which the notions of possible world and epistemic alternative are primitive, 
in the proposed semantics they are non-primitive but are defined from the concept of belief base. 
We show that  this semantics allows us to define the universal epistemic model in a simpler and 
more compact way than  existing  inductive constructions of it. 
We provide (i) a number of semantic equivalence results
for both the basic epistemic
language with  `individual belief'
operators and its extension by 
the notion of 
 `only believing',
 and (ii)  a lower bound complexity result for epistemic logic model checking
 relative to the universal epistemic model.
\end{abstract}

\section{Introduction}

Type spaces were introduced by Harsanyi  \cite{Harsanyi1967} in order to formally represent   higher-order probabilistic
beliefs
of rational players in strategic situations. 
The notion of \emph{universal}
type space
was studied in detail by Mertens \& Zamir \cite{UnivZamir}:
it is defined to be the ``largest'' type space
 which contains all possible states of the world
 as well
 as all possible belief hierarchies of the players.
Mertens \& Zamir showed that, under certain topological assumptions, every
 Harsanyi type space can be mapped to the universal type space by a morphism which preserves
 the
 current state of the world as well as the players' belief hierarchies. 
 Alternative constructions
of the universal type space under different topological assumptions
can be found in \cite{UnivBranden,Heifetz93Univ}. Heifez \& Samet \cite{HeifetzSamet98Univ} proved a variant of Mertens \& Zamir's result
for the general measure-theoretic case.
Battigalli \& Siniscalchi
\cite{BattigalliSiniscalchiUniv}
extended the universal
type space construction
to conditional probabilistic beliefs. More recently, Bjorndahl  \& Halpern
 \cite{BjorndahlHalpernTARK} provided a 
logical analysis of the probabilistic structure
of the universal type space.\footnote{An analysis
of the relationship between Harsanyi's type
spaces and the multi-relational Kripke semantics
of epistemic logic can be found in \cite{LoriniGaleazzi}. }

A qualitative version
of the universal type space
was introduced by 
Fagin et al. \cite{FaginUniv2} (see also \cite{FaginUniv1}).
Specifically, Fagin et al. 
provided 
a
construction
of the ``largest''
(or  ``universal'')
epistemic
model
for the epistemic logic S5$^n$
which contains 
all possible knowledge hierarchies of the agents
in the system.
They studied its relationship
with the standard multi-relational semantics of epistemic logic \cite{Fagin1995}
both for the
basic epistemic language
with 
 `individual knowledge'
 operators
 and for its extension  by common knowledge
 operators.

Both constructions of the universal epistemic model
 in its probabilistic version \`a la Mertens \& Zamir and
 in its qualitative version  \`a la Fagin et al. are inductive.
One has  to define first the set of possible states of the world.
Then, the agents' first-order beliefs about the states of the world
 are defined.
 In the third step, one has to define 
 the agents' second-order beliefs about the states of nature
 and the agents' first-order beliefs,
 and so on so forth. More generally, in order to define an agent's $k+1$-order
 belief
 one has to define first the agents'  $k$-order
 beliefs.

 The general aim of the present paper is to offer
 a simple and  compact  construction
 of the universal
 epistemic model
 that does not require an inductive
 construction
 of the agents' belief hierarchies, as  in 
 Mertens \& Zamir's 
 and
  Fagin et al.'s definitions. 
 Our construction
  of the universal
 epistemic model
  is based on 
 a semantics for  epistemic logic exploiting a belief base abstraction 
 which was recently introduced in  \cite{LoriniAAAI2018}.
 Differently from existing 
 multi-relational Kripke semantics 
  for epistemic logic in which the notions of possible world and epistemic alternative are primitive, in the belief
 base semantics they are non-primitive but are defined from the concept of belief base.
Specifically, in this semantics it is assumed that 
at a given state $s$
agent $i$ considers state $s'$  possible 
if and only if $s'$ satisfies all formulas
that are
included in agent $i$'s belief base at $s$.

The initial motivation for introducing such a semantics 
was to bridge two traditions
that have rarely talked to each other in the past. On the one hand, we have epistemic logic:
it started in the 60s with the seminal work of Hintikka \cite{Hintikka} on the logics of knowledge and belief, it was extended to the multi-agent setting at the end of 80s \cite{Fagin1995,Meyer1995} and
then furtherly developed during the last 20 years, the period of the ``dynamic turn'', with growing research on dynamic epistemic logic \cite{DitmarschHoekKooi07}.
On the other hand, we have syntactic approaches to knowledge representation and reasoning mainly
proposed in the area of artificial intelligence (AI). The latter includes, for instance, work on belief base and knowledge base revision \cite{HanssonJSL,Han99,BDPW02},
belief base merging \cite{KoniecznyPerez},
input-output logic \cite{MakinsonTorreIOL}, as well as more recent work on the so-called
``database perspective'' to the theory of intention by \cite{ShohamJPL}.   All these approaches defend the idea that right level of abstraction for
modeling rational agents  is the ``belief base'' (or ``knowledge base'') level
whereby
the agent is identified with the set of facts that she believes (or knows).

The paper is organized as follows.
In Section \ref{BelBasPre},
we present the belief base semantics for epistemic logic,
as defined in \cite{LoriniAAAI2018,LoriniRomeroAAMAS2019}.
We show how the basic  
language with 
 `individual belief'
 operators
 can be interpreted on this semantics.
 We also show how the semantics can be easily incorporate the
 belief correctness assumption, thereby allowing us to model
 knowledge instead of belief.
 We show that the belief base semantics
 is equivalent to the standard multi-relational
 Kripke semantics,
 in the sense that they lead to the same set
 of validities. Section \ref{univ1}
 is the core of the paper:
 we define the universal epistemic
 model with the help of the belief base semantics. 
As emphasized above,
our definition is simpler and more compact than 
existing definitions
of the universal epistemic model, since it does not require an inductive construction
of the agents' belief hierarchies. 
We show that, as far as 
the basic epistemic language
with
 `individual belief'
 operators is concerned,
 the set
 of validities relative
 to the universal
 epistemic model
 and the set of validities
 relative
 to the generic belief
 base semantics are the same. 
 In Section \ref{extensions},
 we 
 extend the basic epistemic
 language by the notion of
  `only believing'. We show that the equivalence 
  between the two semantics
  does not hold anymore   in the context of this more expressive language:
  the universal epistemic model
  has more validities than the generic
  belief base semantics. 
  Section \ref{univtypespace} is devoted to the comparison
  between Fagin et al.'s inductive construction 
  of the universal
  epistemic model
  and the construction of Section \ref{univ1} exploiting belief bases. 
  In particular,
  we show that,
  from the point of the basic
  epistemic language and of its extension by
  the notion of
    `only believing', there is no difference between the two constructions,
    as they give rise to the same set of validities. 
    In Section \ref{model checking},
    we provide a compact formulation of epistemic
    logic model checking which exploits  the definition
    of    the universal
    epistemic model
    given in Section \ref{univ1}.
    We show that 
    such a compact formulation
    makes epistemic logic model checking
     PSPACE-hard,
    whereas  standard
    epistemic logic model checking
    relative to the multi-relational Kripke semantics
    is polynomial   in the size of the model and the length of the formula.
    In Section \ref{conclude}, we conclude.

\section{Belief base semantics for epistemic logic }\label{BelBasPre}

In this section,
we present a
semantics for epistemic logic
($\logic$)
in which
the accessibility relations 
  in an epistemic model
are not primitive but
they are defined from the primitive
concept
of multi-agent belief base.
The semantics was first introduced in \cite{LoriniAAAI2018}.

\subsection{Multi-agent belief base models}

Assume
a countably  infinite set  of atomic propositions $\PROP = \{p,q, \ldots \}$
and
 a finite set of agents $\AGT = \{ 1, \ldots, n \}$.
 
We define the language $\langminus (\PROP, \AGT)$
by the following grammar in Backus-Naur Form (BNF):
\begin{center}\begin{tabular}{lcl}
  $\alpha$  & $\bnf$ & $ p  \mid \neg\alpha \mid \alpha_1 \wedge \alpha_2  \mid
  \expbel{i} \alpha
                        $
\end{tabular}\end{center}
 where $p$ ranges over $\PROP$
and $i$ ranges over $\AGT$.
 $\langminus(\PROP, \AGT)$
is the language for representing
explicit beliefs of multiple agents.
For simplicity,  we sometimes write
$\langminus$  instead
of $\langminus(\PROP, \AGT)$, when the context is unambiguous.
The formula $\triangle_{i} \alpha$ can be read as ``agent  $i$ explicitly believes that $\alpha$ is true'' or ``$\alpha$ is in agent $i$'s belief base''.
In this language,
we can represent higher-order explicit beliefs,
 for example $\triangle_{i} \triangle_{j} \alpha$ express the fact that agent $i$ explicitly believes that agent $j$ explicitly believes that $\alpha$ is true.

\begin{definition}[Multi-agent belief base]\label{MAB}
A multi-agent belief base  (MBB)  
 is a tuple $ B = (B_1, \ldots, B_n, \states )$
where (i) for every $i  \in \AGT$, $B_i \subseteq\langminus  $
is agent $i$'s belief base,
 and
(ii) $ \states \subseteq \PROP$ is the actual state.
The class of  MBBs  is denoted by $\classbelbasesimple$.
\end{definition}

Formulas
of the language $\langminus$
are interpreted relative to
MBBs as follows.

\begin{definition}[Satisfaction relation for formulas in $\langminus$]\label{satrel1}
Let $ B = (B_1, \ldots, B_n,  \states ) \in \classbelbasesimple$. Then,
the satisfaction relation $\models$
between $B$ and formulas in $ \langminus $
is defined as follows:
\begin{eqnarray*}
B \models p & \Longleftrightarrow & p \in \states \\
B\models \neg \alpha & \Longleftrightarrow &    B \not \models  \alpha \\
B \models \alpha_1 \wedge \alpha_2 & \Longleftrightarrow &    B \models \alpha_1  \text{ and }     B\models \alpha_2 \\
B \models \expbel{i} \alpha   & \Longleftrightarrow & \alpha \in  B_i
\end{eqnarray*}

\end{definition}

%
%
%
%
%

\begin{definition}[Correct MBB]\label{MABc}
Let $B = (B_1, \ldots, B_n, \states ) \in \classbelbasesimple$.
We say that $B$
is correct if and only if, for all $i \in \AGT$
and for all $\alpha \in \langminus$,
if $\alpha \in B_i$ then $B \models \alpha$.
The class of correct MBBs is denoted by
$
\corrclassbelbasesimple$.
\end{definition}

The following definition introduces the
concept
of epistemic alternative.

\begin{definition}[Epistemic alternatives]\label{doxalt}
Let $ B, B' \in \classbelbasesimple$.
Then,  $B \relstate{i} B'$ if and only if, for every $\alpha \in B_i $, $B' \models \alpha$,
where the satisfaction relation $\models$ follows Definition \ref{satrel1}.

\end{definition}
$B \relstate{i} B'$ means
that
$B'$
is an epistemic alternative for agent $i$ at $B$.
 The idea of the previous definition is that
$B'$
is an epistemic alternative for agent $i$ at $B$
if and only if,
$B'$ satisfies all
facts that agent $i$ explicitly believes
at $B$.

A multi-agent belief model (MAB)
is defined to be a
multi-agent belief base
supplemented with
a
set of multi-agent belief bases,
called \emph{context}.
The latter
includes
 all multi-agent belief bases that are compatible with
the agents' common ground \cite{StalnakerCommonGround},
i.e., the body of information that the agents commonly believe to be the case.
\begin{definition}[Multi-agent belief model]\label{MAM}
A multi-agent belief model (MBM) 
is a pair $ (B,\iconstraint)$,
where
$B \in  \classbelbasesimple$
and
 $  \iconstraint \subseteq    \classbelbasesimple$.
The class of  MBMs  is denoted by $\classbelbase $.
\end{definition}

\subsection{Interpretation of epistemic logic language }

Thanks to the epistemic accessibility relations
defined in Definition \ref{doxalt},
we are able to interpret the language of epistemic
logic
relative to MBMs.
Such language
is denoted by $\lang_{\logic}(\PROP, \AGT)$ and
 is defined by the following grammar:
\begin{center}\begin{tabular}{lcl}
  $\phi$  & $\bnf$ & $p  \mid \neg\phi \mid \phi_1 \wedge \phi_2  \mid  \impbel{i} \varphi
                        $\
\end{tabular}\end{center}
where $p$ ranges over $\ATM$
and $i $ ranges over $\Agt$.
For simplicity,  we write
$\lang_{\logic}$  instead
of $\lang_{\logic}(\PROP, \AGT)$, when the context is unambiguous.
The other Boolean constructions  $\top$, $\bot$, $\vee$, $\imp$ and $\eqv$ are defined from $p$, $\neg$ and $\et$ in the usual way.

%
%
%
%

The formula $ \impbel{i}   \varphi$
has to be read
``agent  $i$ implicitly (or potentially)  believes that $\varphi$ is true''
in the sense that agent  $i$ can derive $\varphi$ from her explicit beliefs
(i.e., from the information  in her belief base).
For the sake of simplicity, we sometimes read the formula 
$ \impbel{i}   \varphi$
as ``agent  $i$ believes that $\varphi$''.
We define the dual operator $\impbelposs{i}$
as follows:
\begin{align*}
\impbelposs{i}  \varphi \eqdef  \neg \impbel{i} \neg \varphi.
\end{align*}
$ \Diamond_{i}   \varphi$
has to be read
``$\varphi$ is compatible (or consistent) with agent $i$'s explicit beliefs''.

As usual, for every formula $\varphi \in \lang_{\logic}$
we denote by $\PROP(\varphi)$ the set of atoms in $\PROP$
occurring in $\varphi$.

We are ready
to define what it means for
a multi-agent belief model (MBM) $ (B,\iconstraint) $
to satisfy a formula  $\varphi  $ in $\lang_{\logic}$,
written $ (B,\iconstraint) \models \varphi$.

\begin{definition}[Satisfaction relation for formulas in $\lang_{\logic}$]\label{truthcond2}
Let $ B = (B_1, \ldots, B_n,  \states ) \in \classbelbasesimple$
 and let $ (B,\iconstraint)  \in \classbelbase$. Then:
\begin{eqnarray*}
 (B,\iconstraint) \models p & \Longleftrightarrow & p \in \states \\
  (B,\iconstraint) \models \neg \varphi & \Longleftrightarrow &   (B,\iconstraint) \not \models  \varphi \\
  (B,\iconstraint) \models  \varphi \wedge \psi & \Longleftrightarrow &     (B,\iconstraint) \models  \varphi  \text{ and }    (B,\iconstraint) \models   \psi  \\
 (B,\iconstraint) \models \impbel{i} \varphi & \Longleftrightarrow & \forall B' \in  \iconstraint  : \text{ if } B \relstate{i} B' \text{ then }
 ( B' , \iconstraint) \models \varphi
\end{eqnarray*}
\end{definition}
Note that, according to the last clause,
an agent $i$ implicitly believes that $\varphi$
(i.e., $\impbel{i} \varphi $) if and only if $\varphi$
is true at all states 
that are compatible with the information in $i$'s belief base.

Figure \ref{fig:framework} illustrates
the general idea
behind the belief base semantics,
especially for what concerns the relationship
between the agents' belief
bases and the agents' common ground (or context)
and the relationship between
the latter and the agents' implicit beliefs.
While an agent's belief
base captures the agent's private information,
the common ground captures the agents' public information.
An agent's implicit belief corresponds to a fact that the agent can deduce
from the public information and her private information.

 \begin{figure}[h]
\centering
\includegraphics[scale=0.48]{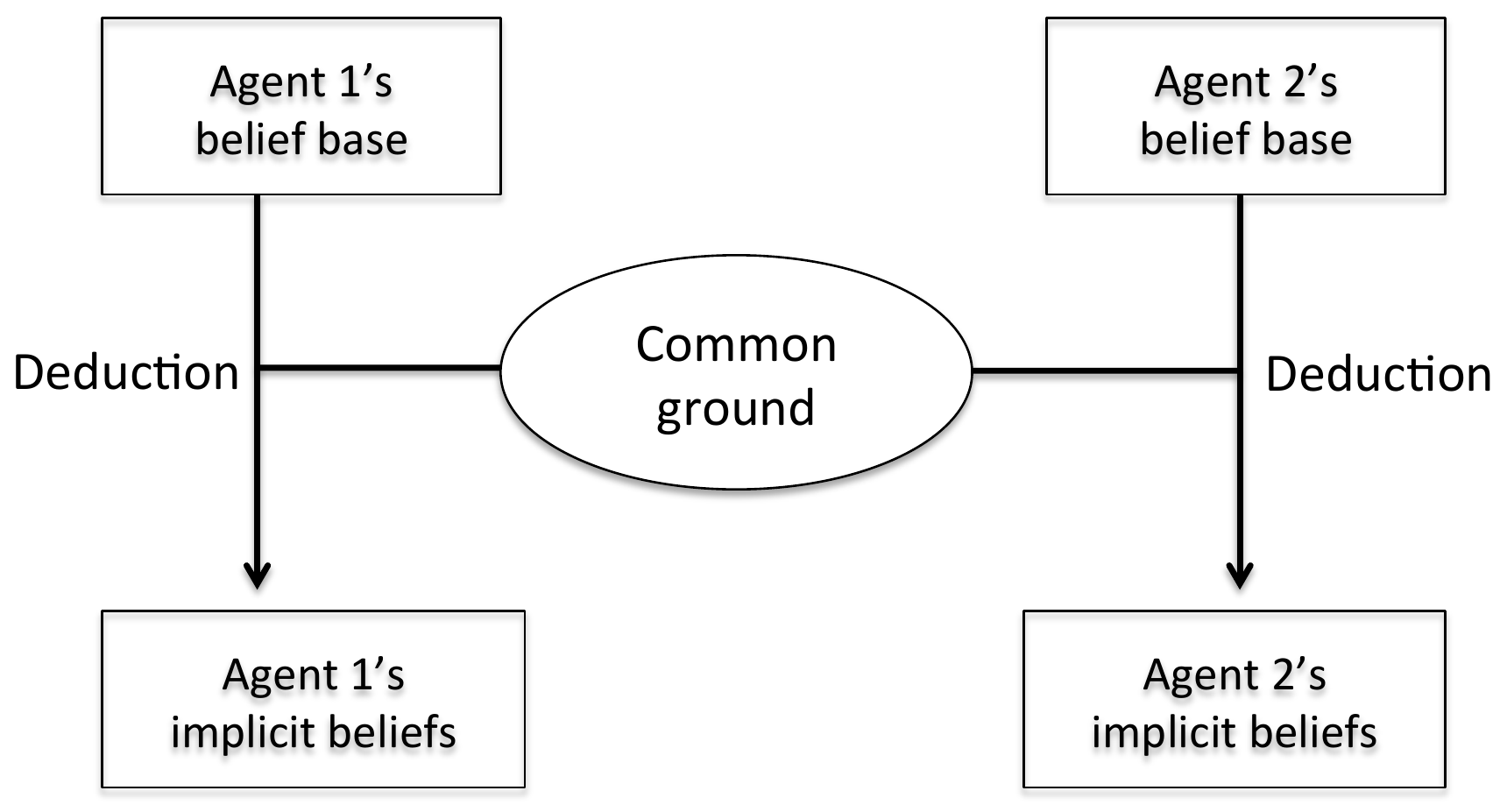}
\caption{Conceptual framework}
\label{fig:framework}
\end{figure}

Given a formula
$\varphi$
in  $\lang_{\logic}$,
we denote by  $\mathit{depth}(\varphi)$
its model depth which is defined as follows:
\begin{align*}
\mathit{depth}(p) & =0\\
\mathit{depth}(\neg \varphi) & =\mathit{depth}( \varphi) \\
\mathit{depth}(\varphi \wedge \psi) & = \max\big(\mathit{depth}(\varphi ), \mathit{depth}( \psi)\big) \\
\mathit{depth}(\impbel{i} \varphi) & =\mathit{depth}( \varphi) +1
\end{align*}

\subsection{Model classes and validity}

%
%

In some situations,
it may be useful
to assume
that
agents' beliefs are correct, i.e.,
what an agent believes
is true.
When talking about correct (or true) explicit and implicit beliefs,
it is usual to call them
explicit and implicit
knowledge.
Indeed, we assume that the terms
  ``true belief'',
      ``correct belief''
and   ``knowledge''
are  synonyms.
The following definition introduces belief
correctness for multi-agent belief models.
\begin{definition}[Belief correctness]\label{BCprop}
Let $ (B,\iconstraint) \in \classbelbase $.
We say that $  (B,\iconstraint)$
satisfies belief correctness (BC) if and only if
$B\in \iconstraint $
and,
for every
 $B' \in  \iconstraint$, $B'  \relstate{i} B'$.
\end{definition}

As the following proposition highlights, belief correctness for multi-agent
belief models is completely characterized by the fact that the actual world is included in the agents' common ground and that the agents' explicit beliefs are correct, i.e.,
if an agent has $\alpha$
in her belief base then $\alpha$
is true in the actual state of the world.
\begin{proposition}
Let $ (B,\iconstraint) \in \classbelbase $.
Then, $  (B,\iconstraint)$ satisfies BC if and only if $B\in \iconstraint $ and
$B ' \in \corrclassbelbasesimple$
for all $B' \in \iconstraint $.
\end{proposition}

We let  $\classbelbase_{BC}$
denote
the class of MBMs 
satisfying property $BC$.

Let $\varphi \in \lang_\logic$,
we say that $\varphi$
is valid
for the class  $\classbelbase$, denoted by $\models_{\classbelbase}\varphi$,
if and only if, for every  $(B , \iconstraint) \in \classbelbase$,
we have $ (B , \iconstraint) \models \varphi $.
We say that
$\varphi$
is satisfiable
for the  class  $\classbelbase$
if and only if $\neg \varphi $
is not
valid
for the class $\classbelbase$.
Satisfiability
and validity of a formula $\varphi$
in $ \lang_\logic$ relative to the class $\classbelbase_{BC}$ (denoted by $\models_{\classbelbase_{BC}} \varphi$)
are defined in an analogous way.

\subsection{Equivalence with Kripke semantics}\label{Kripke}

In this section,
we show that the belief base semantics
for the epistemic language 
$\lang_{\logic}$
defined in the previous section
is equivalent to the traditional
multi-relational Kripke semantics for
epistemic logic \cite{Fagin1995}.

\begin{definition}[Multi-relational Kripke model]\label{KripkeStructures}
A multi-relational Kripke model is a structure $M=(W, \Rightarrow_1, \ldots , \Rightarrow_n, \omega )$
such that $W $
is a set of states,  $\Rightarrow_i \subseteq W \times W$ is agent $i$'s
epistemic accessibility relation, and $\omega : \PROP \longrightarrow 2^W$
is a valuation function.
Multi-relational Kripke models
satisfying reflexivity  
are those for which every relation $\Rightarrow_i$ is reflexive
(i.e., for every $w \in W$, $w \Rightarrow_i w$).
\end{definition}

The interpretation of formulas 
in $\lang_{\logic}$
 relative to
a  multi-relational Kripke model
$M=(W, \Rightarrow_1, \ldots , \Rightarrow_n, \omega  )$
and a world $w$ in $W$ is defined as follows:
\begin{align*}
(M,w) \models p & \text{ iff } w \in \omega(p)\\
(M,w) \models \neg \varphi  & \text{ iff } (M,w) \not \models \varphi \\
(M,w) \models  \varphi \wedge \psi & \text{ iff } (M,w)  \models \varphi \text{ and } (M,w)  \models \psi \\
(M,w) \models \impbel{i} \varphi  & \text{ iff }  \forall v \in W: \text{ if } w \Rightarrow_i v  \text{ then } (M,v)  \models \varphi
\end{align*}
Notions
of validity and satisfiability relative
to the class of multi-relational Kripke models
are defined in the usual way. 
We denote the fact that $\varphi$
is valid relative to the class
of multi-relational Kripke models 
(resp. multi-relational Kripke models satisfying reflexivity) 
by $\models_{\mathbf{KripkeM}} \varphi$
(resp. $\models_{\mathbf{ReflKripkeM}} \varphi$).

The equivalence 
between the belief base
semantics for the language $\lang_\logic$
and the multi-relational Kripke semantics
is stated by the following theorem.
\begin{theorem}\label{KripkeThe}
Let $\varphi \in \lang_{\logic} $.
Then, 
\begin{itemize}
\item $\models_{\classbelbase} \varphi$
if and only if  $\models_{\mathbf{KripkeM}} \varphi$,
\item $\models_{\classbelbase_{BC}} \varphi$
if and only if  $\models_{\mathbf{ReflKripkeM}} \varphi$.
\end{itemize}
\end{theorem}

\section{Universal context}\label{univ1}

\newcommand{  \ucontext}{\classbelbasesimple_{\top}}

 \newcommand{  \corrucontext}{\corrclassbelbasesimple_{\top}}

We here define the notion of \emph{universal context},
as the context containing all possible explicit belief
hierarchies for the agents.
\begin{definition}[$\alpha$-context]\label{Induced}
Let  $\alpha \in \langminus$.
We define
\begin{align*}
\classbelbasesimple_\alpha= \{B\in \classbelbasesimple \suchthat B \models \alpha\}
\end{align*}
to be the $\alpha$-context. The universal context is simply $ \classbelbasesimple_\top$
which is the same thing as $ \classbelbasesimple $.
\end{definition}
The formula $\alpha$
in the previous definition corresponds
to the concept
of \emph{integrity constraint},
as traditionally defined in the area of
knowledge representation \cite{ReiterIntegrity,KoniecznyPerez}.

It is worth to consider a more specific
notion of universal context
under the assumption of belief
correcteness.

\begin{definition}[$\alpha$-context  with belief correctness]\label{Induced2}
Let  $\alpha \in \langminus$.
We define
\begin{align*}
\corrclassbelbasesimple_\alpha= \{B\in \corrclassbelbasesimple \suchthat B \models \alpha\}
\end{align*}
to be the $\alpha$-context with belief correctness. The universal context with belief correctness is simply $ \corrclassbelbasesimple_\top$ which is the same thing
as  $ \corrclassbelbasesimple$.
\end{definition}

Let
$\varphi \in \lang_\logic$, 
we say that $\varphi$
is valid
relative to the universal context $\ucontext$, denoted by $\models_{\ucontext}\varphi$,
if and only if, for every  $B \in \ucontext$,
we have $ (B , \ucontext) \models \varphi $.
Analogously, 
we say that $\varphi$
is valid
relative to the universal context satisfying  BC, denoted by $\models_{ \corrucontext}\varphi$,
if and only if, for every  $B \in  \corrucontext$,
we have $ (B ,  \corrucontext) \models \varphi $.

We say that
$\varphi$
is satisfiable
relative to the universal context  $\ucontext$ (resp. $ \corrucontext$)
if and only if $\neg \varphi $
is not
valid
relative to the universal context $\ucontext$ (resp. $ \corrucontext$).

The following theorem highlights that 
the set of validities 
for the language $ \lang_\logic$
relative to the universal context 
$\ucontext$ (resp. $ \corrucontext$)
and the set of validities relative to the class
of models $\classbelbase$
(resp. $\classbelbase_{ \mathit{BC}}$)
are the same. In other words, 
the language $ \lang_\logic$
is not able to distinguish
the universal context including 
all possible hierarchies of the agents' beliefs
from
incomplete models
in which some belief hierarchy may be missing.

\begin{theorem}\label{teoremone}
Let
$\varphi \in \lang_\logic$.
Then,
\begin{itemize}
\item  $\models_{\classbelbase} \varphi$ if and only if $\models_{\ucontext}\varphi$,
\item  $\models_{\classbelbase_{ \mathit{BC}}} \varphi$ if and only if $\models_{ \corrucontext}\varphi$.
\end{itemize}
\end{theorem}

The following theorem is
a corollary of 
Theorems
\ref{KripkeThe}
and \ref{teoremone}.

\begin{theorem}\label{corollario}
Let
$\varphi \in \lang_\logic$.
Then,
\begin{itemize}
\item   $\models_{\ucontext}\varphi$ if and only if $\models_{\mathbf{KripkeM}} \varphi$,
\item   $\models_{ \corrucontext}\varphi$
if and only if $\models_{\mathbf{ReflKripkeM}} \varphi$.
\end{itemize}
\end{theorem}

The notion of universal context allows us to easily
define the notion of model with \emph{maximal  general uncertainty}, that is, to say
a model in which the agents have maximal uncertainty about
the state of the world as well as maximal uncertainty about
the agents' $k$-order beliefs
for every $k \geq 1$, 
 where the order of a belief is defined inductively as follows: 
 (i) an agent's belief is order 1 if and only if its content is a 
propositional  formula that does not mention beliefs; (ii) an agent's belief is order $k+1$ if and only it is a belief about 
the
 agents' $k$-order beliefs.

\begin{definition}[MBM with maximal general uncertainty]\label{MaxIgn}
Let $B=(B_1 , \ldots, B_n, \states) \in \classbelbasesimple $
such that $B_i= \emptyset $
for every $i \in \AGT$.
Then, $(B,  \ucontext)$ is called
 MBM with maximal general uncertainty, while  $(B,  \corrucontext)$ is called 
 MBM with   maximal  general uncertainty and belief correctness.
\end{definition}

Note that in our semantics maximal general uncertainty coincides
with the fact that the agents' belief bases are empty (i.e., the agents do not know anything). 
The following
proposition is a direct consequence  of  Theorem \ref{teoremone}.
\begin{proposition}\label{propone}
Let $\varphi \in \lang_{\logic}$, let $i \in \AGT$
and let $(B,  \ucontext)$ (resp.  $(B,  \corrucontext)$)
be a  MBM with maximal general uncertainty (resp.  a MBM with maximal general uncertainty and belief correctness).
Then,
\begin{itemize}
\item $\varphi$ is satisfiable for the class $\classbelbase$ if and only if $ (B,  \ucontext) \models \impbelposs{i} \varphi $,
\item $\varphi$ is satisfiable for the class $\classbelbase_{\mathit{BC}}$ if and only if $ (B,  \corrucontext) \models  \impbelposs{i} \varphi $.
\end{itemize}
\end{proposition}

It highlights the essential aspects of models with maximal general uncertainty. 
To see this, suppose $\varphi $
and $\neg \varphi$
are both satisfiable for the class $\classbelbase$
(resp. for the class $\classbelbase_{\mathit{BC}}$).
Then, by Proposition \ref{propone}, 
we have 
$ (B,  \ucontext) \models \impbelposs{i} \varphi  \wedge \impbelposs{i}  \neg \varphi  $
(resp. $  (B,  \corrucontext) \models \impbelposs{i} \varphi  \wedge \impbelposs{i}  \neg \varphi  $), 
where
$\impbelposs{i} \varphi  \wedge \impbelposs{i}  \neg \varphi $
captures agent $i$'s uncertainty about $\varphi$.
More generally,
for every formula
$\varphi$
yielding information either about the state
of the world
or about the agents' higher-order beliefs,
if $\varphi$
and 
 $\neg \varphi$ are both satisfiable, 
 then
 in a model with maximal uncertainty
  every agent $i$
 has uncertainty about $\varphi$.

\section{Only believing}\label{extensions}

The aim of this section is to show that if
we increase the expressive power 
of our multimodal epistemic
language by the notions of `believing at most'
and `only believing', then the equivalence result
between the generic semantics
in terms
of MBMs 
and the universal context semantics 
of Section \ref{univ1} does not hold anymore.

Let $\lang_{\logicext}(\PROP, \AGT)$
be the language which 
extends  language $\lang_{\logic}$
by modal operators of implicitly believing `at most'
of the form $ \complbel{i}$,
where $\logicext$
stands for ``Extended Epistemic Logic''.
It is defined by the following grammar:
 \begin{center}\begin{tabular}{lcl}
  $\phi$  & $\bnf$ & $p  \mid \neg\phi \mid \phi_1 \wedge \phi_2  \mid  \impbel{i} \varphi \mid  \complbel{i}\varphi
                        $\
\end{tabular}\end{center}
where $p$ ranges over $\PROP$
and $i $ ranges over $\Agt$.
For simplicity,  we write
 $\lang_{\logicext}$  instead
of $\lang_{\logicext}(\PROP, \AGT)$, when the context is unambiguous.

The formula $\complbel{i}\varphi$
has to be read ``agent $i$ believes at most that $\neg \varphi$''
and has the following interpretation
relative to MBMs.

\begin{definition}[Satisfaction relation (cont.)]\label{truthcond3}
Let\\ $ B = (B_1, \ldots, 
B_n,  \states ) \in \classbelbasesimple$
 and let $ (B,\iconstraint)  \in \classbelbase$. Then:
\begin{eqnarray*}
 (B,\iconstraint) \models \complbel{i}\varphi & \Longleftrightarrow & \forall B' \in   \big( \iconstraint \setminus \relstate{i} (B) \big)  : 
 ( B' , \iconstraint) \models \varphi
\end{eqnarray*}
where $ \relstate{i} (B) = \{B' \suchthat B \relstate{i}  B'\}$.
\end{definition}

The definition of modal depth of
a formula in   $\lang_{\logicext}$
extends the definition
of the modal depth
of formulas in $\lang_{\logic}$
by the following additional clause:
\begin{align*}
\mathit{depth}( \complbel{i}\varphi) & =\mathit{depth}( \varphi) +1
\end{align*}

By combining the operators $\impbel{i}$
and $\complbel{i}$ in the appropriate way,
we can reconstruct the universal modality \cite{Hemaspaandra96,GorankoPassyUMOD}:
\begin{align*}
\mathsf{U}\varphi \eqdef \impbel{i} \varphi \wedge \complbel{i}\varphi 
\end{align*}
where $\mathsf{U}\varphi $
has to be read  ``$\varphi$
is universally true''.
We can moreover
the reconstruct the ``only believing'' modality \cite{LevesqueOnly,HalpernLakeOnly,Lakemeyer93AK}:
\begin{align*}
\mathsf{O}_i \varphi \eqdef \impbel{i} \varphi \wedge \complbel{i}\neg \varphi 
\end{align*}
where $\mathsf{O}_i \varphi $
has to be read  ``all that agent $i$ believes is $\varphi$''.
From the perspective of the logic of only believing,
all that agent $i$ believes is $\varphi$
if and only if agent $i$ believes at least that $\varphi$
is true
(i.e., $\impbel{i} \varphi $)
and she believes at most that $\varphi$ is true
(i.e., $\complbel{i}\neg \varphi  $).

Definitions of validity and satisfiability
for formulas in $\lang_{\logicext}$
relative to the class $\classbelbase$ (resp. $\classbelbase_{\mathit{BC}}$)
and relative to the universal context $\ucontext$ (resp. $\corrucontext$)
coincide with those for formulas in $\lang_{\logic}$ given above.

The following theorem highlights
that the semantics for the language $\lang_\logicext$
based on the universal context
contains more validities than the semantics
for the language $\lang_\logicext$
based on  the generic class
$\classbelbase$.

\begin{theorem}\label{nonequiv}
We have the following relationship between sets of validities:
\begin{itemize}
\item    $\{\varphi \in \lang_\logicext \suchthat \models_{\classbelbase} \varphi \} \subset \{\varphi \in \lang_\logicext \suchthat \models_{\ucontext}\varphi \}$,
\item    $\{\varphi \in \lang_\logicext \suchthat \models_{\classbelbase_{\mathit{BC}}} \varphi \} \subset \{\varphi \in \lang_\logicext \suchthat \models_{\corrucontext} \varphi \}$.
\end{itemize}
\end{theorem}
Showing that 
$\{\varphi \in \lang_\logicext \suchthat \models_{\classbelbase} \varphi \} \subseteq \{\varphi \in \lang_\logicext \suchthat \models_{\ucontext}\varphi \}$
is trivial since $(B,\ucontext) \in \classbelbase$ for every $B \in \ucontext $.  
Furthermore,
it is  easy to find a formula $\chi$ such that $ \not \models_{\classbelbase} \chi $
and $\models_{\ucontext}\chi$. The following formula is an example:
\begin{align*}
 \chi \eqdef \mathsf{E} \big(\bigwedge_{p \in X} p \wedge \bigwedge_{q \in Y} \neg q \big) 
\end{align*}
with $\mathsf{E} \psi \eqdef \neg \mathsf{U} \neg \psi $
and finite $X,Y \subseteq \PROP$
such that $X \cap Y= \emptyset$.

Note that, we can easily adapt the proof of Theorem \ref{KripkeThe}
to show that 
the set of 
$\lang_\logicext$-validities relative to the class of multi-relational
Kripke models
and the set of validities relative to the generic belief base semantics
are the same.
\begin{theorem}\label{KripkeThe2}
Let $\varphi \in \lang_{\logicext} $.
Then, 
\begin{itemize}
\item $\models_{\classbelbase} \varphi$
if and only if $\models_{\mathbf{KripkeM}} \varphi$,
\item $\models_{\classbelbase_{BC}} \varphi$
if and only if $\models_{\mathbf{ReflKripkeM}} \varphi$.
\end{itemize}
\end{theorem}
The following is a direct consequence of Theorems 
\ref{nonequiv} and
\ref{KripkeThe2}.

\begin{theorem}\label{nonequiv2}
We have the following relationship between sets of validities:
\begin{itemize}
\item    $\{\varphi \in \lang_\logicext \suchthat \models_{\mathbf{KripkeM}} \varphi \} \subset \{\varphi \in \lang_\logicext \suchthat \models_{\ucontext}\varphi \}$,
\item    $\{\varphi \in \lang_\logicext \suchthat \models_{\mathbf{ReflKripkeM}} \varphi \} \subset \{\varphi \in \lang_\logicext \suchthat \models_{\corrucontext} \varphi \}$.
\end{itemize}
\end{theorem}

\section{Qualitative belief structures}\label{univtypespace}

In this section, we 
consider a different representation
of the universal epistemic
model containing all possible belief hierarchies for
the agents in the system
first proposed by \cite{FaginUniv2} (see also \cite{FaginUniv1}).
Differently from
the definition of the universal
context given in Section \ref{univ1}, which does not require any
inductive construction of belief hierarchies, 
Fagin et al.'s
definition of the universal 
epistemic model  is inductive. 
As emphasized by \cite{KleinBenthem},
Fagin et al.'s universal
epistemic model can be seen as qualitative counterpart of
the notion of probabilistic universal type space
by \cite{UnivZamir}.
A similar inductive construction of the universal
epistemic
model
was proposed more recently
by \cite{DBLP:conf/kr/BelleL10} (see also \cite{DBLP:conf/ijcai/AucherB15}),
as a semantics
for the logic of multi-agent only knowing. 

Fagin et al. study
the universal
epistemic model
for the multimodal logic S5$^n$,
including both positive and negative introspection principles
for knowledge. We here consider variants
of Fagin et al.'s construction  for the multimodal logics
K$^n$
and 
KT$^n$,
the latter including the veracity
principle for knowledge
according to which what an agent knows cannot be false.\footnote{Belle \& Lakemeyer \cite{DBLP:conf/kr/BelleL10} 
offer an inductive construction
of the universal epistemic model
for the multimodal logic K45$^n$.}

Given an arbitrary set $Y$,
let
$F(Y) $
be the set of functions with domain $Y$
and codomain $ \{0,1\}$.
We define the set $Z_k$
in an inductive way as follows:
\begin{align*}
Z_0= &F(\ATM)\\
Z_{k +1}=& Z_k \times F(Z_k)^n.
\end{align*}
Thus, we have
\begin{align*}
Z_{k +1}=& Z_0  \times F(Z_0)^n \times \ldots \times F(Z_k)^n.
\end{align*}
Elements of 
$Z_0$
are denoted by $f_0$,
while 
elements of
$F(Z_k)^n$
are denoted by $f_{k+1}, f_{k+1}', \ldots$
Moreover,  for every $i \in \AGT$ and for every $f_{k+1} \in F(Z_k)^n$,
$f_{k+1}(i)$ denotes the $i$-th component in the tuple 
 $f_{k+1} $.

Elements of $Z_{k }$
are called
 \emph{$k+1$-ary worlds},
 or simply  \emph{$k+1$-worlds}.
 Elements of $\bigcup_{k \in  \mathbb{N}_0} Z_{k }$
 are called  \emph{worlds}.
 The set of worlds  is denoted by $\mathbf{W}$.

The set of belief
structures 
is defined as follows:
\begin{align*}
\mathbf{BS}  =& Z_0 \times F(Z_{0})^n \times F(Z_{1})^n \times \ldots
\end{align*}
Thus, a belief structure 
is a countably infinite sequence $f= (f_0, f_1, \ldots)$
such that $f_0 \in Z_0 $
and, for every $k \in \mathbb{N}_0$,
$f_{k+1} \in F(Z_k)^n$.
It follows that 
$f= (f_0, f_1, \ldots)$
is a belief structure 
if and only if, for every $k \in \mathbb{N}_0$,
 $(f_0, \ldots, f_k )$
 is a 
$k+1$-world.

\begin{definition}[Coherent world and belief structure]\label{CoherCond}
A world 
 $(f_0, \ldots, f_{k} )$
is said to be coherent if and only if,
the following conditions hold:
\begin{itemize}

\item for every $i \in \AGT$
and for every $2 \leq h \leq k$:\\ if $ f_{h}(i) (g_0, \ldots, g_{h-1})=1$
then $ f_{h-1}(i) (g_0, \ldots, g_{h-2})=1$;

\item for every $i \in \AGT$
and for every $1 \leq h \leq k-1$:\\
if $ f_{h}(i) (g_0, \ldots, g_{h-1})=1$
then there exists $g_h \in F(Z_{h-1})^n$
such that 
$ f_{h+1}(i) (g_0, \ldots, g_{h-1}, g_{h})=1$.

\end{itemize}
 The set of coherent worlds  is denoted by $\mathbf{CW}$.
We say that 
the belief structure  $f= (f_0, f_1, \ldots)$
is coherent if,
for every $k \in \mathbb{N}_0$,
 $(f_0, \ldots, f_k )$
 is a 
coherent world.
The set of coherent belief structures  is denoted by $\mathbf{CBS}$.

\end{definition}

The following 
definition introduces the property of correctness for worlds and belief structures.

\begin{definition}[Correct world and belief structure]\label{CorrectCond}
A world 
 $(f_0, \ldots,  f_{k} )$ with $k \geq 1$
satisfies belief correctness if and only if, for every $i \in \AGT$,
$f_k(i)(f_0, \ldots, f_{k-1})=1 $.
 The set of coherent  worlds satisfying belief correctness (BC)  is denoted by $\mathbf{CW}_{BC}$.
We say that 
the belief structure  $f= (f_0, f_1, \ldots)$
is correct if,
for every $k \geq 1$,
 $(f_0, \ldots, f_k )$
 is a 
correct world.
The set of coherent belief structures satisfying belief correctness (BC)  is denoted by $\mathbf{CBS}_{BC}$.
\end{definition}


A formula $\varphi$ of the 
language $\lang_{\logicext}$
is interpreted relative to
a coherent  $k+1$-world
 $(f_0, \ldots, f_k ) \in \mathbf{CW}$ 
 such that $k \geq \mathit{depth}(\varphi)$,
 as follows:
 \begin{eqnarray*}
(f_0, \ldots, f_k)  \models p & \Longleftrightarrow & f_0(p) =1\\
(f_0, \ldots, f_k)  \models \neg \varphi & \Longleftrightarrow & (f_0, \ldots, f_k) \not  \models \varphi\\
(f_0, \ldots, f_k)  \models  \varphi  \wedge \psi & \Longleftrightarrow & (f_0, \ldots, f_k)   \models \varphi \text{ and }(f_0, \ldots, f_k)  \models \psi\\
(f_0, \ldots, f_k)   \models \impbel{i} \varphi & \Longleftrightarrow & \forall (g_0, \ldots, g_{k-1}) \in  Z_{k-1} :\\
&& \text{if }  f_k(i) (g_0, \ldots, g_{k-1}) = 1 \\
&&\text{then }
(g_0, \ldots, g_{k-1})   \models \varphi \\
(f_0, \ldots, f_k)   \models \complbel{i}\varphi & \Longleftrightarrow & \forall (g_0, \ldots, g_{k-1})  \in  Z_{k-1} :\\
&& \text{if }  f_k(i) (g_0, \ldots, g_{k-1}) = 0 \\
&&\text{then }
(g_0, \ldots, g_{k-1})   \models \varphi 
\end{eqnarray*}

The following proposition is a generalization of \cite[Lemma 2.5]{FaginUniv2}
to the language $\lang_{\logicext}$.\footnote{\cite[Lemma 2.5]{FaginUniv2} only applies
to the standard epistemic language with epistemic operators $\impbel{i}$ and without
operators $ \complbel{i}$.}

\begin{proposition}\label{lemmuccio}
Let $ \mathit{depth}(\varphi) = h$,
$k \geq h$, and
$(f_0, \ldots, f_{h})  ,
(f_0, \ldots, f_{k}) \in \mathbf{CW} $.
Then,
$(f_0, \ldots, f_{h})   \models \varphi $
if and only if 
$(f_0, \ldots, f_{k})   \models \varphi $. 

\end{proposition}

Let 
$\varphi\in\lang_{\logicext}$
and 
let  $f= (f_0, f_1, \ldots) \in \mathbf{CBS} $
be a coherent belief structure.
We say that $f$ satisfies $\varphi$,
denoted by $f \models \varphi$
if $(f_0, \ldots, f_h) \models \varphi$
where $h = \mathit{depth}(\varphi)$.
It is worth noting that, by Proposition \ref{lemmuccio},
if
$k \geq h$ 
then 
$f \models \varphi$
if and only if $(f_0, \ldots, f_k) \models \varphi$. 
We moreover say that $\varphi$
is valid relative to the class $\mathbf{CBS} $,
denoted by $\models_{\mathbf{CBS}} \varphi$,
if $f \models \varphi$
for every $f \in \mathbf{CBS}$.
We say that 
$\varphi$
is satisfiable relative to the class $\mathbf{CBS}$
if $\not \models_{\mathbf{CBS}}  \neg \varphi$.
Definitions of validity  relative to the class 
 $\mathbf{CBS}_{BC}$  (denoted by $\models_{\mathbf{CBS}_{BC}} \varphi$)
 and satisfiability 
 are defined analogously.

It is straightforward to adapt the proof of \cite{FaginUniv2}[Theorem 3.1]
in order to prove the following equivalence result
between the coherent belief structure semantics
and the Kripke semantics relative to 
the language  $\lang_{\logic}$.

 \begin{theorem}\label{equivKripke2}
Let
$\varphi \in \lang_{\logic}$.
Then,
\begin{itemize}
\item   $\models_{\mathbf{CBS}} \varphi$
if and only if $ \models_{\mathbf{KripkeM}}\varphi$, 
\item      $\models_{\mathbf{CBS}_{BC}} \varphi$
if and only if  $\models_{\mathbf{ReflKripkeM}}\varphi$.
\end{itemize}
 \end{theorem}
 
 The following theorem
 is a corollary of Theorems \ref{corollario} and
 \ref{equivKripke2}.
 
 \begin{theorem}\label{corollario2}
Let
$\varphi \in \lang_{\logic}$.
Then,
\begin{itemize}
\item   $\models_{ \ucontext}\varphi$ if and only if $\models_{\mathbf{CBS}} \varphi$, 
\item       $\models_{ \corrucontext}\varphi$ if and only if $\models_{\mathbf{CBS}_{BC}} \varphi$.
\end{itemize}
 \end{theorem}

 The following  theorem
 strengthens  Theorem \ref{corollario2}
 by stating that
 the set of $\lang_{\logicext}$-validites
 relative to the universal
 context $\ucontext$ (resp. $ \corrucontext$)
 is the same as the set of $\lang_{\logicext}$-validites relative to the class $\mathbf{CBS}$
 (resp. $\mathbf{CBS}_{BC}$).

 \begin{theorem}\label{teoremiccolo}
Let
$\varphi \in \lang_{\logicext}$.
Then,
\begin{itemize}
\item   $\models_{ \ucontext}\varphi$ if and only if $\models_{\mathbf{CBS}} \varphi$, 
\item      $\models_{ \corrucontext}\varphi$ if and only if $\models_{\mathbf{CBS}_{BC}} \varphi$.
\end{itemize}
 \end{theorem}

 More generally, 
 not even 
the  epistemic language
 $ \lang_{\logicext}$ ---
 which extends  the language $ \lang_{\logic}$
by the  notions
of
`believing at most'
and `only believing' ---
can distinguish
 the 
 belief structure
 semantics  \`a la Fagin et al.
 from the universal
 context
 semantics
exploiting the belief base abstraction.

\section{Model checking}\label{model checking}

\newcommand{\uncertain}[1] {\bigcirc_{#1}   }

The notions of $\alpha$-context
and universal context 
defined in Section  \ref{univ1} (Definition \ref{Induced})
allow us to offer a  compact formulation of the model
checking problem for the formulas in the language $\lang_\logic$.\footnote{ Other
compact formulations of the epistemic logic
model checking have been proposed in the literature.
For instance
 \cite{LomuscioRaimondi2015} provide
 a semantics for epistemic logic  model checking based on the concept of interpreted
 system, while the approach
 by  \cite{BenthemEijckGattingerSu2015,DBLP:conf/aamas/HoekIW12,DBLP:conf/kr/CharrierHLMS16} builds
 on propositional observability. 
  }

\smallskip

{\centering
\fbox{
\begin{minipage}{20em}
\noindent \underline{$\alpha$-context model checking}\\
\emph{Given}: $\varphi \in \lang_\logic$,
$\alpha \in   \langminus$ and a
finite $ B  \in \classbelbasesimple$.\\
\noindent\emph{Question}: Do we have
$ (B,\classbelbasesimple_\alpha ) \models \varphi $?
\end{minipage}
}

}

\smallskip
\noindent where the multi-agent belief base (MBB) $ B=(B_1, \ldots, B_n, \states)$
is said to be finite  if $\states$
and every $B_i$
are finite.

Note that, if $\alpha= \top$, then model checking consists
in verifying whether $\varphi$
is true at a given finite MBB of the universal
context, that is, in verifying whether $ (B,\ucontext ) \models \varphi$
for a specific finite $B \in \ucontext$. 
%

The model checking problem
for   multi-agent epistemic
logic
K$^n$
interpreted relative
to the standard multi-relational Kripke semantics
of Definition \ref{KripkeStructures}
is known to be P-complete
with respect to the size of the input
formula to be checked and the size of the model \cite{GradelOtto}.
In this section, 
we are going to show that the
previous 
compact formulation of the model checking problem 
for the formulas in  $\lang_\logic$ is PSPACE-hard,
the same complexity as the corresponding satisfiability problem.
Our result highlights that  the gains in terms of
compactness
of the model checking problem
are counterbalanced by losses on the complexity side.

We provide a polynomial reduction of true quantified boolean formulas (TQBF)
to model checking for formulas in  $\lang_\logic$.

Let 
us assume that a quantified boolean formula $\chi$ is in (not necessarily closed) prenex
normal form and can be written as follows:
\begin{align*}
\lambda . \pi
\end{align*}
where $ \pi$
is a propositional formula
and $\lambda$
is a (possibly empty)
sequence
 $Q_0 p_0 \ldots Q_m p_m$
such that $\{p_0, \ldots, p_m\} \subseteq \PROP(\pi) $
and $Q_0, \ldots, Q_m \in \{\exists, \forall\}$, with $\PROP(\pi)$
denoting the set of atoms occurring in $\pi$.
If $\{p_0, \ldots, p_m\}= \PROP(\pi) $
then $\chi$
is said to be \emph{closed},
as it has no free variables.
The \emph{length}
of $\chi$
is defined to be
the length of the sequence  $p_0 \ldots p_m$.
For notational convenience,
when $\lambda$
is non empty,
we write $\lambda[k]$
to denote the propositional
variable
in the $k$-th
position
in the sequence
 $p_0 \ldots p_m$.

Let us suppose that the language of quantified boolean formulas
is
built
over the set of atomic variables $\PROP$.
Moreover,
let us define the following translation
from the language of quantified boolean formulas
to the language $\lang_\logic(\PROP,\{1\})$:
\begin{align*}
\mathit{tr}(p)  =& \text{ for } p \in \PROP  \\
\mathit{tr}\big( \forall p \lambda. \pi  \big)  = &
[\![\lambda]\!] \mathit{tr}(\lambda. \pi)
\end{align*}
where
\begin{align*}
 & [\![\lambda]\!] \psi \eqdef \impbel{1} (   \uncertain{1}\lambda[0]  \rightarrow \psi ) \\
  \end{align*}
  and
  \begin{align*}
   \uncertain{1}\lambda[0]  \eqdef & \impbelposs{1} \lambda[0] \wedge \impbelposs{1} \neg \lambda[0] & \text{ if } \lambda \text{ is  non-empty}  \\
   \uncertain{1}\lambda[0] \eqdef & \top & \text{ otherwise }
\end{align*}
and where the Boolean connectives are translated homomorphically.
We define the dual operator of $[\![\lambda]\!]$
as follows:
\begin{align*}
\langle\!\langle\lambda\rangle\!\rangle \psi \eqdef  \neg [\![\lambda]\!] \neg \psi.
  \end{align*}

For every
quantified boolean
formula $\chi=\lambda. \pi = Q_0 p_0 \ldots Q_m p_m . \pi$ and
for every $\states \subseteq \ATM$,
we define:
\begin{align*}
\Sigma_{\chi,\states}  =&  \big(  ( \PROP(\pi)\setminus \{ \lambda[0]  \}  ) \cap    \states     \big) \cup \\
&  \{  \neg q \suchthat q \in (\PROP(\pi)\setminus \{ \lambda[0]  \}   )   \setminus \states      \} \cup \\
  & \bigcup_{1 \leq k \leq m } \{  \expbel{1}^{k-1} \bigwedge_{q \in \PROP(\pi)\setminus \{ \lambda[k]  \}   }  (\triangle_{1} q \vee  \triangle_{1} \neg q)          \}
\end{align*}
where $\expbel{1}^{k-1} \alpha$
is inductively defined as follows:
\begin{align*}
\expbel{1}^{k-1} \alpha & \eqdef \alpha & \text{ if } k = 1 \\
\expbel{1}^{k-1} \alpha & \eqdef \expbel{1}\expbel{1}^{k-2}\alpha &  \text{ if } k >1
\end{align*}

Moreover,
let us define
 \begin{align*}
B^{(\chi,\states)}  =  \big( B_1^{(\chi,\states)} ,   \states  \big)
 \end{align*}
where
\begin{align*}
B_1^{(\chi,\states)}  & = \Sigma_{\chi,\states} & \text{ if } \lambda  \text{ is non-empty}\\
B_1^{(\chi,\states)}  &= \emptyset &  \text{ otherwise}
  \end{align*}

The following lemma
is crucial  for proving
the main result
of this section
about complexity of model checking  
for formulas in $\lang_\logic$.

\begin{lemma}\label{intermediatelemma2}
Let  $\states \subseteq \PROP$
and
let
$\chi=\lambda. \pi$ be a quantified boolean formula.
Then,
$\states \models \chi$
if and only if
$\big( B^{(\chi,\states)}, \ucontext \big) \models \mathit{tr}( \chi) $.
\end{lemma}

Thanks to   Lemma \ref{intermediatelemma2}
and
 the fact that
the TQBF problem (i.e., the problem
of checking whether a quantified
boolean
formula is true in a given valuation)
is PSPACE-hard  \cite{StockmeyerQBF},
we have a PSPACE-hardness result
for model checking of $\lang_\logic$-formulas.
Indeed,
let $\chi= Q_0 p_0 \ldots Q_m p_m . \pi$
be a quantified boolean formula. Then, 
the sizes of
$B_1^{(\chi,\states)}$ and $ \mathit{tr}(\chi) $
are polynomial  in $m$.
Therefore, by Lemma \ref{intermediatelemma2}
we
have a polynomial reduction
of the TQBF problem
into
$\logic$-model checking.
Since the former
is PSPACE-hard,
then the latter is PSPACE-hard too.
\begin{theorem}\label{Th1}
The $\alpha$-context model checking problem
for formulas in $\lang_\logic$ is PSPACE-hard.
\end{theorem}

From \cite{DBLP:journals/ai/HalpernM92} we know that satisfiability
checking 
for $\lang_\logic$-formulas
relative to multi-relational 
Kripke models
of Definition \ref{KripkeStructures} is PSPACE-complete.
Thus, thanks to Theorem \ref{corollario},
we can conclude that
satisfiability
checking 
for $\lang_\logic$-formulas
relative to the universal context $\ucontext$
is PSPACE-complete as well.
Consequently, Theorem \ref{Th1}
also shows that
the compact version of
model checking for 
$\lang_\logic$-formulas
exploiting the notion
of $\alpha$-context is at least as complex
as  satisfiability  checking for
$\lang_\logic$-formulas.

\section{Conclusion}\label{conclude}

We have offered 
a novel representation
of the universal
epistemic
model
which does not require an inductive
construction
of the agents' belief hierarchies.
Our solution
exploits the notion
of belief
base, a natural
abstraction for
representing epistemic
attitudes
of rational players in interactive situations,
widely used in the area of knowledge representation and reasoning (KR). 
Our lower bound  complexity result
for model checking relative
to the universal
epistemic model
clearly displays the trade-off
between the compactness of the model checking problem
representation
and its computational complexity. 

Directions of future research are manifold. As emphasized in Section \ref{univtypespace},
Fagin et al.'s original definition of 
the universal
epistemic model
applies  to the epistemic logic S5$^n$
which includes
principles of 
positive and negative introspection for knowledge
(i.e., if an agent knows that $\varphi$/does not know that $\varphi$,
then she knows that she knows that $\varphi$/does not know that $\varphi$).
The logics we presented in this paper do not make 
any assumption about introspection for belief or knowledge.
Future work will be devoted to extend our analysis
of the universal epistemic
model
 to variants of epistemic
logic with introspection. 
We also plan to extend our comparative
analysis of the different
semantics for epistemic logic 
to more expressive epistemic languages
including common belief and distributed belief operators.
Last but not least,
we plan to explore the connection between  
the representation
of agents' belief hierarchies
using the belief
base abstraction
and  type spaces with finite depth of reasoning
as defined in \cite{KetsTypes}.
Indeed, the belief base abstraction allows to naturally define
the notion of $k$-level (resource-bounded) reasoner, as an agent whose belief base contains explicit beliefs
of at most  order $k$. This is clearly related to the idea of type spaces with finite depth of reasoning
\`a la Kets.

\section*{Acknowledgements}

Support from the
ANR project CoPains 
(``Cognitive Planning in Persuasive Multimodal Communication'')
	and the
 ANR-3IA Artificial and Natural Intelligence Toulouse Institute is gratefully acknowledged.
 
\bibliographystyle{eptcs}

\appendix

\section{Proofs}

This appendix presents a selection
of the proofs of the technical
results presented in the paper.

\subsection{Proof of Theorem \ref{KripkeThe}}

\begin{proof}
We only prove the first item. The proof of the second item is similar.

In order to prove the right-to-left direction
we show that if $\varphi$
is satisfiable
for $\classbelbase$
then it is satisfiable for the class of multi-relational Kripke models.
Let $(B, \iconstraint)\in\classbelbase $ such that $(B, \iconstraint) \models \varphi $.
We build the multi-relational Kripke model $M=(W, \Rightarrow_1, \ldots , \Rightarrow_n, \omega )$
corresponding to $(B, \iconstraint) $
as follows:
\begin{itemize}
\item $W= \{w_{B'} \suchthat B' \in \{B\}  \cup \iconstraint \}$,
\item for every $i \in \AGT$ and for every $w_{B'},w_{B''} \in W$, $w_{B'} \Rightarrow_i w_{B''}$ iff  $B'  \relstate{i} B''$,
\item for every $p \in \ATM$, $\omega(p) = \{w_{B'}  \in W \suchthat  B' \models p \}$,
\end{itemize}
 Clearly, we have $(M, w_B) \models \varphi$
 iff $(B, \iconstraint ) \models \varphi $.
 Thus, $(M, w_B) \models \varphi$ 
 since we supposed that $(B, \iconstraint ) \models \varphi $.
 
 As for the left-to-right direction suppose 
 $\varphi$
 is satisfiable for  the class of 
 multi-relational Kripke models. We know that the
 multimodal logic K$^n$
 interpreted relative to 
  multi-relational Kripke models has the finite model property.
  Consequently, there exists 
a \emph{finite} multi-relational Kripke model $M=(W, \Rightarrow_1, \ldots , \Rightarrow_n, \omega )$
 and $w \in W$ such that  $(M, w) \models \varphi$.
 Let $\mathit{name}: W \rightarrow \PROP \setminus \PROP(\varphi)$
 be an injective function.
 Such an injection  exists since
 the set $W$ is finite and 
 the
  the set $\PROP$
 is assumed to be infinite. 
 We define the context $\iconstraint= \{B^v \suchthat v \in W \}$
where, for every $v \in W$ and for every $i \in \AGT$,
\begin{align*}
B^v_i= \{ \bigvee_{u \in W \suchthat v \Rightarrow_i u }  \mathit{name}(u)  \},
\end{align*}
and $\states^v= \big( \omega(v) \cap \ATM(\varphi) \big) \cup \{\mathit{name}(v)\}$.
By induction on the structure of $\varphi$, it is routine task to show that $(B^w, \iconstraint )\models \varphi$
iff $(M,w) \models \varphi$.
Thus,  $(B^w, \iconstraint )\models \varphi$ since $(M,w) \models \varphi$.
\end{proof}

\subsection{Proof of Theorem \ref{teoremone}}\label{appsec}

\begin{proof}
We prove the first item,
as the proof of the second item is analogous.

Left-to-right direction is obvious. 
As for the right-to-left direction,
we prove that if $\varphi$
is satisfiable relative to the class $\classbelbase$,
then it is satisfiable relative to the universal context.

Let $(B, \iconstraint)\in\classbelbase $ such that $(B, \iconstraint) \models \varphi $.
We build the multi-relational Kripke model $M=(W, \Rightarrow_1, \ldots , \Rightarrow_n, \omega )$
corresponding to $(B, \iconstraint) $
as follows:
\begin{itemize}
\item $W= \{w_{B'} \suchthat B' \in \{B\}  \cup \iconstraint \}$,
\item for every $i \in \AGT$ and for every $w_{B'},w_{B''} \in W$, $w_{B'} \Rightarrow_i w_{B''}$ iff  $B'  \relstate{i} B''$,
\item for every $p \in \ATM$, $\omega(p) = \{w_{B'}  \in W \suchthat  B' \models p \}$,
\end{itemize}
where the interpretation of formulas 
 relative to
a  multi-relational Kripke model
$M=(W, \Rightarrow_1, \ldots , \Rightarrow_n, \omega )$
and a world $w$ in $W$ was defined in Section \ref{Kripke}.
Clearly, we have $(M, w_B) \models \varphi$
iff $(B,  \iconstraint) \models \varphi $. Thus, 
we have
$(M, w_B) \models \varphi$ since $(B,  \iconstraint) \models \varphi $.

In what follows, for notational convenience we denote elements of $W$
by $w, v , u , \ldots$

We define the filtration of $M$.
Let $\Sigma \subseteq \lang_\logic$
be an arbitrary finite set of formulas which is closed under subformulas. (Cf. Definition 2.35 in \cite{Bla01} for
a definition of subformulas closed set of formulas.)
Let the equivalence relation $\equiv_\Sigma$ on $W$ be defined as follows. For all $w,v \in W$:
\begin{align*}
w \equiv_\Sigma v \text{ iff } \forall \varphi \in \Sigma : (M,w) \models \varphi \text{ iff } (M,v) \models \varphi.
\end{align*}
Let $|w |_\Sigma$
be the equivalence class of the world $w$ with respect to the equivalence relation $\equiv_\Sigma$.

We define $W_\Sigma$ to be the filtrated set of worlds with respect to $\Sigma$:
\begin{align*}
W_\Sigma = \{ |w |_\Sigma \suchthat w \in W \}.
\end{align*}
Clearly, $W_\Sigma$
is a finite set.

Let us define the filtrated valuation function $\omega_\Sigma $. For every $p \in \PROP$,
we define:
\begin{align*}
\omega_\Sigma (p) & = \{  |w |_\Sigma \suchthat (M,w) \models p  \}  & \text{if } p \in \PROP(\Sigma)  \\
\omega_\Sigma (p)  & = \emptyset &  \text{otherwise}
\end{align*}
with $\PROP(\Sigma)  = \bigcup_{\psi \in \Sigma} \PROP(\psi) $.

Finally,
for every $i \in \AGT$,
we define
agent $i$'s accessibility relation as follows:
\begin{align*}
\Rightarrow_{i,\Sigma  }= \{ (|w |_\Sigma  , |v |_\Sigma  ) \suchthat |w |_\Sigma  , |v |_\Sigma \in W^\Sigma \text{ and } w  \Rightarrow_i v  \}.
\end{align*}

Let us define:
\begin{align*}
\Rightarrow_\Sigma  = \bigcup_{i \in \AGT} \Rightarrow_{i,\Sigma  }.
\end{align*}


The model
$(W_\Sigma , \Rightarrow_{1,\Sigma  },\ldots, \Rightarrow_{n,\Sigma  }, \omega_\Sigma )$
is the smallest filtration of $M$
under $\Sigma$.
Let $ M_{\{\varphi\}}= (W_{\{\varphi\}} , \Rightarrow_{1,\{\varphi\}},\ldots,\Rightarrow_{n,\{\varphi\}}, \omega_{\{\varphi\}} )$.
Clearly, we have $(M_{\{\varphi\}},  |w_B |_{\{\varphi\}}   ) \models \varphi$
since $(M, w_B) \models \varphi$.

In what follows, for notational convenience we denote elements of $W_{\{\varphi\}}$
by $x ,y, \ldots$

The next step of the proof consists in unraveling the finite model $ M_{\{\varphi\}}$
up to the modal depth of $\varphi$.
 
 Let
 \begin{align*}
\mathit{Seq}(\varphi)= \{(x_0, \ldots, x_{k} ) \suchthat 
k \leq \mathit{depth}(\varphi) \text{ and }
x_0, \ldots, x_{k} \in W_{\{\varphi\}}\}
\end{align*}
be the set of sequences
of worlds in $W_{\{\varphi\}}$
of length at most $\mathit{depth}(\varphi)$.
Elements of $\mathit{Seq}(\varphi)$
are denoted by $\overrightarrow{x}, \overrightarrow{y}, \ldots$
The length of the sequence $\overrightarrow{x}
$ is denoted by $\mathit{length}(\overrightarrow{x})$.
For every $0 \leq k \leq \mathit{length}(\overrightarrow{x})$,
$\overrightarrow{x}[k]$
denotes the $k$-th element in $\overrightarrow{x}$
while $ \overrightarrow{x}[\mathit{last}]$
denotes the last element in $ \overrightarrow{x}$.
Let $\overrightarrow{x}=(x_0, \ldots, x_{k})$,
we write $ \overrightarrow{x}.y$
to denote the sequence 
$(x_0, \ldots, x_{k} , y)$.
 
We define
 the tree-like multi-relational Kripke model $M'=(W', \Rightarrow_1', \ldots , \Rightarrow_n', \omega')$
 as follows:
 
 \begin{itemize}
 \item $W'= \{\overrightarrow{x}\in \mathit{Seq}(\varphi)
 \suchthat \overrightarrow{x}[0]= |w_B |_{\{\varphi\}} \text{ and } \forall 0 \leq k < \mathit{length}(\overrightarrow{x}), \overrightarrow{x}[k]  \Rightarrow_{\{\varphi\}} \overrightarrow{x}[k+1]  \}$,
 
 \item
 for every $i \in \AGT$
 and for every $\overrightarrow{x}, \overrightarrow{y} \in W'$,
  $\overrightarrow{x} \Rightarrow_i'  \overrightarrow{y}$ iff $\exists y \in W_{\{\varphi\}}$ such that $ \overrightarrow{y}= \overrightarrow{x}.y $
 and $ \overrightarrow{x}[\mathit{last}] \Rightarrow_{i,\{\varphi\}} y$,

 \item for every $p \in \ATM$, $\omega '(p) = \{\overrightarrow{x} \in W' \suchthat   \overrightarrow{x}[\mathit{last}] \in \omega_{\{\varphi\}} (p) \}$.

  \end{itemize}
  
  Clearly, the tree-like model $M'$
  is finite.

 By induction on the structure of $\varphi$,
 it is routine exercise to check that $(M', \overrightarrow{ x_0} ) \models \varphi$
 iff \\$(M_{\{\varphi\}},  |w_B |_{\{\varphi\}} ) \models \varphi$,
with $\overrightarrow{ x_0}  = (  |w_B |_{\{\varphi\}})$.
Thus,  $(M', \overrightarrow{ x_0} ) \models \varphi$ since 
$(M_{\{\varphi\}},  |w_B |_{\{\varphi\}} ) \models \varphi$.

Now, let us denote by $\mathit{Term}_{M'}$
the set of terminal nodes in the tree $M'$.
That is, $\mathit{Term}_{M'} = \{ \overrightarrow{x} \in W' \suchthat
\forall  \overrightarrow{y} \in W' \text{ , } \overrightarrow{x} \not \Rightarrow' \overrightarrow{y}  \}$
with $\Rightarrow' = \bigcup_{i \in \AGT} \Rightarrow_i'$.

We define $k$-level explicit mutual belief that $\alpha$ in an inductive way as follows:
  \begin{align*}
  \mathsf{EB}^0\alpha & \eqdef \alpha\\
     \mathsf{EB}^{k+1}\alpha &\eqdef           \mathsf{EB} \ \mathsf{EB}^{k}\alpha\\
          \mathsf{MB}^{k}\alpha &\eqdef        \bigwedge_{0 \leq h \leq k }   \mathsf{EB}^h \alpha 
    \end{align*}
    where $ \mathsf{EB} \alpha \eqdef \bigwedge_{i \in \AGT} \expbel{i} \alpha $.

We define the labelling function $L$
over nodes in $W'$ as follows:
  \begin{align*}
  L( \overrightarrow{x}) = & \bigwedge_{p \in \mathit{Val}(\varphi , \overrightarrow{x})}  p \wedge   \bigwedge_{p \in  \ATM(\varphi) \setminus \mathit{Val}(\varphi , \overrightarrow{x})} \neg p  
  \wedge 
  \bigwedge_{i \in \AGT} \expbel{i} \bot \\
  & \text{ if }  \overrightarrow{x} \in \mathit{Term}_{M'}  \\
    L( \overrightarrow{x}) = & \bigwedge_{p \in \mathit{Val}(\varphi , \overrightarrow{x})}  p \wedge  \bigwedge_{p \in  \ATM(\varphi) \setminus \mathit{Val}(\varphi , \overrightarrow{x})} \neg p  
    \wedge   \bigwedge_{i \in \AGT} \expbel{i} \bigvee_{ \overrightarrow{y} \suchthat \overrightarrow{x} \Rightarrow_i'  \overrightarrow{y} } L( \overrightarrow{y})\\
 &   \text{ if }  \overrightarrow{x} \in W' \setminus \mathit{Term}_{M'}  
    \end{align*}
where
$
\mathit{Val}(\varphi , \overrightarrow{x})= \{p \in \ATM(\varphi)\suchthat (M', \overrightarrow{x} ) \models p \}
$
is the set of atoms in $\varphi$
which are true at $\overrightarrow{x}$.
Note that every $L( \overrightarrow{x})$
is a (finitary) formula of the language $\langminus$.

The labelling function is used to construct,
for every $i\in \AGT$,
the 
$\langminus$-formula which corresponds to the ``tree'' of valuations  rooted in $ (  |w_B |_{\{\varphi\}})$.

Let $\overrightarrow{ x_0}  = (  |w_B |_{\{\varphi\}})$ and
let $B' = (B_1', \ldots, B_n', \states')$
such that:

\begin{itemize}
\item for every $i \in \AGT$, $B_i' =
\bigcup_{ p \in \PROP \setminus \PROP(\varphi) } \{  \mathsf{MB}^{\mathit{depth}(\varphi)} \neg p \} \cup
 \{ \bigvee_{ \overrightarrow{y} \in W' \suchthat \overrightarrow{ x_0} \Rightarrow_i'  \overrightarrow{y} } L( \overrightarrow{y})  
  \}$,
\item $\states' = \{p \suchthat \overrightarrow{ x_0} \in \omega '(p) \}$.
\end{itemize}

Furthermore, let us define the (universal) multi-relational Kripke model $M''=(W'', \Rightarrow_1'', \ldots , \Rightarrow_n'', \omega'' )$
corresponding to the universal context $\ucontext$
as follows: (i) $W''= \{w_{B''} \suchthat B'' \in \ucontext \}$,
(ii) for every $i \in \AGT$ and for every $w_{B''},w_{B'''} \in W$, $w_{B''} \Rightarrow_i w_{B'''}$ iff  $B''  \relstate{i} B'''$,
and (iii) for every $p \in \ATM$, $\omega(p) = \{w_{B''}  \in W \suchthat  B''\models p \}$.

By construction of $B' $,
we can show that $(M'', w_{B'} )$
and $(M', \overrightarrow{ x_0} )$
with $\overrightarrow{ x_0}  = (  |w_B |_{\{\varphi\}})$
are $\mathit{depth}(\varphi)$-bisimilar \cite[Definition 2.30]{Bla01}. Thus, 
$(M'', w_{B'} ) \models \varphi$
since
 $(M', \overrightarrow{ x_0} ) \models \varphi$.
 It follows that 
 $(B',  \ucontext) \models \varphi$,
 since clearly $(M'', w_{B'} ) \models \varphi$ iff $(B',  \ucontext) \models \varphi$.
\end{proof}

\subsection{Proof of Proposition \ref{lemmuccio}}

\begin{proof}
The proof is by induction on the structure of the formulas. 
Boolean cases are trivial.
The case $\varphi=\impbel{i} \psi$
is proved in the same way as \cite[Lemma 2.5]{{FaginUniv2}}.
Let us prove the case $\varphi=\complbel{i} \psi$. 
Assume that 
$k \geq h= \mathit{depth}(\varphi) $
and
$(f_0, \ldots, f_{k})   \models\complbel{i} \psi$.
Since $\varphi=\impbel{i} \psi$, we have $h \geq 1$.
Moreover, let $(g_0, \ldots, g_{h-1}) \in Z_{h-1}$
such that $f_h(i)(g_0, \ldots, g_{h-1})=0$.
By Definition \ref{CoherCond},
for all $(g_0, \ldots, g_{k-1}) \in Z_{k-1}$,
 $f_k(i)(g_0, \ldots, g_{k-1})=0$.
Since $(f_0, \ldots, f_{k})   \models\complbel{i} \psi$,
it follows by definition
that 
$(g_0, \ldots, g_{k-1}) \models  \psi$
for all $(g_0, \ldots, g_{k-1}) \in Z_{k-1}$.
Thus, by induction hypothesis, 
$(g_0, \ldots, g_{h-1}) \models \psi$.
It follows that, for every 
$(g_0, \ldots, g_{h-1}) \in Z_{h-1}$
such that $f_h(i)(g_0, \ldots, g_{h-1})=0$,
$(g_0, \ldots, g_{h-1}) \models \psi$.
Thus, 
$(f_0, \ldots, f_{h})   \models\complbel{i} \psi$.
The proof of the other direction is similar.
\end{proof}

\subsection{Proof of Theorem \ref{teoremiccolo}}

\begin{proof}
We only prove the first item. The second item can be proved in a similar way.

Let $\varphi \in \lang_\logicext$.
We will prove that $\varphi $
is satisfiable relative to 
$ \ucontext $
iff $\varphi$
is satisfiable relative to
the class $\mathbf{CBS}$.

We build the multi-relational Kripke model $M=(W,  \relstate{1}, \ldots ,  \relstate{n},  \relstate{1}^c, \ldots ,  \relstate{n}^c, \omega )$
corresponding to the $\PROP(\varphi)$-restriction of the universal context $ \ucontext $ as follows:
\begin{itemize}
\item $W= \{ B' \in  \ucontext  \suchthat  \forall q \in \PROP\setminus \PROP(\varphi)
\text{ and } \forall k \geq 0:(B',  \ucontext) \models   \mathsf{MB}^{k} \neg q \}  $,
\item for all $i \in \AGT$ and for all ${B'},{B''} \in W$, ${B'}  \relstate{i}^c {B''}$ iff  $B''  \not \in \relstate{i} (B')$,
\item for all $p \in \ATM$, $\omega(p) = \{B' \in W \suchthat  B' \models p \}$,
\end{itemize}
where the $k$-level
mutual belief operator $ \mathsf{MB}^{k}$
is defined as in the proof of Theorem  \ref{teoremone} (Section \ref{appsec}).
In what follows, for notational convenience we denote elements of $W$
by $w, v , u , \ldots$
Interpretation of $\lang_\logicext$-formulas
relative to $M$
and to a world $w \in W$
is defined as in Section \ref{Kripke}
for boolean formulas and for formula $\impbel{i} \varphi $.
We add the following clause for the $\complbel{i}$-operator:
\begin{align*}
(M,w) \models  \complbel{i} \varphi  & \text{ iff }  \forall v \not  \in  \relstate{i}(w):  (M,v)  \models \varphi
\end{align*}
We leave to the reader the task of checking that
 $\varphi$
is satisfiable relative to 
$ \ucontext $ iff $M$ satisfies $\varphi$.


Similarly, we build the multi-relational Kripke model $M'=(W', \mathcal{T}_1, \ldots , \mathcal{T}_n, \mathcal{T}_1^c, \ldots , \mathcal{T}_n^c, \omega' )$
corresponding to the $\PROP(\varphi)$-restriction of the set of coherent belief structures  
as follows:
\begin{itemize}
\item $W'=   \{ f \in  \mathbf{CBS} \suchthat  \forall q \in \PROP\setminus \PROP(\varphi),
f_0(q)=0
\text{ and if } g_0(q)=1 \text{ then } f_k(i) (g_0, \ldots, g_{k-1})=0, \forall k \geq 0, \forall i \in \AGT, \forall g \in \mathbf{CBS} \} $,
\item for all $i \in \AGT$ and for all ${f},{g} \in W'$, $f   \mathcal{T}_i g$ iff \\
$f_{k}(i)(g_0, \ldots, g_{k-1})=1$  for all $k > 1$,
\item for all $i \in \AGT$ and for all ${f},{g} \in W'$, $f  \mathcal{T}_i^c g$ iff  $g \not \in   \mathcal{T}_i (f)$,
\item for all  $p \in \ATM$, $\omega'(p) = \{f  \in W' \suchthat  f_0(p)=1 \}$.
\end{itemize}
Interpretation of $\lang_\logicext$-formulas
relative to $M'$
and to a world $f \in W'$
is defined as usual.
We have the following truth conditions for the
 operators $\impbel{i}$
and  $ \complbel{i}$:
\begin{align*}
(M,w) \models  \impbel{i} \varphi  & \text{ iff }  \forall v \in W: \text{ if } w   \mathcal{T}_i v  \text{ then } (M,v)  \models \varphi\\
(M,w) \models  \complbel{i} \varphi  & \text{ iff }  \forall v \in W: \text{ if } w   \mathcal{T}_i^c v  \text{ then } (M,v)  \models \varphi
\end{align*}
We leave to the reader the task of checking that
$\varphi$
is satisfiable relative to 
$  \mathbf{CBS} $ iff $M'$ satisfies $\varphi$.

For every world $w\in W$ in the Kripke model $M$,
we build a coherent belief structure $f^w = (f_0^w, f_1^w, \ldots) $. We define $f_0^w$ to be the function such that, for every $p\in \PROP$,  $f_0^w(p)=1 $
iff $w\in \omega(p)$. Moreover, suppose $f_0^w, \ldots, f_k^w$ have been defined
for each $w \in W$. Then,
we define $f_{k+1}^w$
to be the function such that, 
 for every $i \in \AGT$,
$f_{k+1}^w(i)^{-1} (\{1\})= \{  (g_0^v, \ldots, g_k^v) \suchthat w  \Rightarrow_i  v \}$
where $f_{k+1}^w(i)^{-1} (\{1\})$
is the inverse image  by $f_{k+1}^w(i)$ of the subset $\{1\}$ of the codomain $ \{0,1\}$.

Let us define the mapping $\tau: w \mapsto f^w$
from $W$ to $W'$.
The following intermediate proposition
will be useful for the rest of the proof.
\begin{proposition}\label{intermprop}
Let $f \in W'$. Then, there exists $v \in W$
such that $\tau(v)=f$.
\end{proposition}
\begin{proof}

Let 
$f= (f_0, f_1, \ldots) \in W'$.
We show how to construct $v \in W$
such that $\tau(v)=f$.

Let us define:
\begin{align*}
\beta_{i, (f_0, f_1)} \eqdef \bigvee_{ g_0 \in Z_0 \suchthat f_1(i) (g_0)=1}
g_0(\varphi)
\end{align*}
where
\begin{align*}
g_0(\varphi)\eqdef \bigwedge_{ p\in \ATM(\varphi) \suchthat g_0(p)=1} p
\wedge  \bigwedge_{ p\in \ATM(\varphi)  \suchthat g_0(p)=0} \neg p.
\end{align*}
 Moreover, for every $k > 1$,
 let us define:
 \begin{align*}
\beta_{i, (f_0, \ldots, f_k)} \eqdef & \bigvee_{ (g_0, \ldots , g_{k-1}) \in Z_{k-1} \suchthat f_k(i) (g_0, \ldots , g_{k-1})=1}
\big( g_0(\varphi) \wedge \\
&\bigwedge_{j \in \AGT} \expbel{j} \beta_{j, (g_0, \ldots, g_{k-1})} \big).
\end{align*}
Note that $\beta_{i, (f_0, f_1)}$ and every $\beta_{i, (f_0, \ldots, f_k)}$
are (finitary) formulas of $\langminus$.

We define the multi-agent belief base
 $B^f= (B_{1}^f, \ldots, B_{n}^f, \states^f) $
as follows:
 \begin{align*}
B_{i}^f & = \bigcup_{k \geq 1 }\{\beta_{i, (f_0, \ldots, f_k)}\} \cup \bigcup_{k \geq 0, q \in \PROP \setminus \PROP(\varphi) }\{\mathsf{MB}^{k} \neg q \}  ,\\
\states^f &= \{p \in \ATM \suchthat f_0(p)=1 \}.
\end{align*}
Clearly, $ B^f$
belongs to $W'$.
Moreover,
it is the case that 
the multi-agent belief base
$B_{i}^f$
and the belief structure $f$ define the same belief hierarchy.
Specifically, we have $\tau(B^f)=f $.
 \end{proof}
By means of  Proposition \ref{intermprop},
we can show that 
the relation $ \{(w, f) \in W \times W' \suchthat \tau(w) =f \}  $
is a bisimulation
between the Kripke models $M$
and $M'$.
Therefore,
suppose $(M, v) \models \varphi$.
It follows that $(M' , \tau(v)) \models \varphi $.
Viceversa, suppose $(M' , f) \models \varphi $.
Then, by Proposition  \ref{intermprop},
there exists $v \in W$
such that $\tau(v) =f$.
Since $v$
and $f$
are bisimilar, we have $(M, v) \models \varphi$.
 \end{proof}

\subsection{Proof of Lemma \ref{intermediatelemma2}}

\begin{proof}

In order to prove Lemma \ref{intermediatelemma2},
we first prove the following intermediate result.

\begin{lemma}\label{intermediatelemma1}
Let $\states \subseteq \PROP$
and
let
$\chi=\lambda. \pi$ be a quantified boolean formula.
Moreover,
let $B' = (B_1',  \states'  ), B''\\ =(B_1'',  \states''  ) \in \classbelbasesimple$
such that
$\big( \states \cap \PROP(\pi) \big) =  \big( \states'  \cap \PROP(\pi) \big)= \big( \states''  \cap \PROP(\pi) \big) $,
 $B_1^{(\chi,\states)} \subseteq B_1' $,
  $B_1^{(\chi,\states)} \subseteq B_1'' $,
$(B', \ucontext ) \models \uncertain{1}\lambda[0]$
and $(B'', \ucontext ) \models \uncertain{1}\lambda[0]$.
  Then, $(B', \ucontext ) \models \mathit{tr}( \chi) $
  iff $(B'', \ucontext ) \models \mathit{tr}( \chi) $.
\end{lemma}

\begin{proof}
Let
$B' = (B_1',  \states'  )$
and $B'' = (B_1'',  \states''  )$.
We assume
$\big(\states \cap \PROP(\pi) \big)= \big( \states'  \cap \PROP(\pi) \big)= \big( \states'' \cap \PROP(\pi)\big)$,
 $B_1^{(\chi,\states)} \subseteq B_1' $,
  $B_1^{(\chi,\states)} \subseteq B_1''$,
$(B', \ucontext ) \models \uncertain{1}\lambda[0]$
and
$(B'', \ucontext ) \models \uncertain{1}\lambda[0]$.

The proof of the lemma is by induction on the length of the formula
$\chi$.

\paragraph{Base case}
Let  $\lambda$
be the empty sequence.
Clearly,
$B' \models \pi$
  iff $B'' \models \pi $,
  since $\big(\states' \cap \PROP(\pi) \big)= \big( \states''  \cap \PROP(\pi) \big)  $.
  Therefore,
$\big(B', \ucontext \big) \models \pi $
  iff $\big(B'', \ucontext \big) \models \pi $.

\paragraph{Inductive case}
 We prove that the statement
 is true for sequence $\lambda$ of length $k+1$, if
 we suppose it is true
for sequence $\lambda$ of length $k$.
Suppose $\lambda$
has length $k+1$.
Therefore,
we can assume that $\lambda = Q p \lambda'$.
We assume  $Q = \forall$.
The proof for $Q = \exists$ is analogous.

Let
\begin{align*}
\mathcal{R}_{[\![\lambda']\!]} =  \{ (B''',B'''') \in \relstate{1}
 \suchthat (B'''',  \ucontext) \models    \uncertain{1}\lambda'[0] \}.
\end{align*}
Clearly, for all $B''' \in   \ucontext$,
we have
\begin{align*}
 (B''',  \ucontext) \models  [\![\lambda']\!] \psi  \text{ iff }  (B'''',  \ucontext) \models   \psi  \text{ for all }  B'''' \in \mathcal{R}_{[\![\lambda']\!]} (B''').
\end{align*}

Now,
suppose $(B', \ucontext ) \models \mathit{tr}( \forall p \lambda'. \pi ) $.
The latter is equivalent to
$(B', \ucontext ) \models [\![\lambda']\!] \mathit{tr}(\lambda'. \pi)$.
The latter is equivalent to
\begin{align*}
& (B''',  \ucontext) \models   \mathit{tr}(\lambda'. \pi)  \text{ for all }  B''' \in \mathcal{R}_{[\![\lambda']\!]} (B').
\end{align*}
For all $B''' \in \big( \mathcal{R}_{[\![\lambda']\!]} (B') \cup \mathcal{R}_{[\![\lambda']\!]} (B'') \big)$,
we clearly have:
\begin{align*}
& (B''',  \ucontext) \models  \uncertain{1}\lambda'[0] .
\end{align*}
since $B_1^{(\chi,\states)} \subseteq B_1' $ and
  $B_1^{(\chi,\states)} \subseteq B_1''$.
Moreover,
by the initial assumption that
$\big(\states \cap \PROP(\pi) \big)=\big( \states'  \cap \PROP(\pi)\big)=\big(  \states'' \cap \PROP(\pi)\big)$,
 $B_1^{(\chi,\states)} \subseteq B_1' $ and
  $B_1^{(\chi,\states)} \subseteq B_1''$,
for all $B''' \in  \big( \mathcal{R}_{[\![\lambda']\!]} (B') \cup \mathcal{R}_{[\![\lambda']\!]} (B'') \big)$
we have:
\begin{align*}
& B_1^{(\lambda'. \pi,\states \cup \{p \}) }\subseteq B_1''' ;\\
& B_1^{(\lambda'. \pi,\states \setminus \{p \})} \subseteq B_1''' ;\\
&  \big( (\states \cup \{p \}) \cap \PROP(\pi) \big) = \big( \states'''  \cap \PROP(\pi) \big) \text{ or } \\
&  \big( (\states \setminus \{p \}) \cap \PROP(\pi)  \big)= \big( \states'''  \cap \PROP(\pi)\big).
\end{align*}

Furthermore, by  the fact that
$(B', \ucontext ) \models \uncertain{1}p$
and
$(B'', \ucontext ) \models \uncertain{1}p $,
we have that:
\begin{align*}
& \forall  B'''' \in  \mathcal{R}_{[\![\lambda']\!]} (B''), \exists B''' \in \mathcal{R}_{[\![\lambda']\!]} (B') \\
& \text{such that }
\big( \states''''\cap \PROP(\pi) \big)=\big( \states'''  \cap \PROP(\pi)\big).
\end{align*}

Therefore, by induction hypothesis and the fact that
$(B''',  \ucontext) \models   \mathit{tr}(\lambda'. \pi)  \text{ for all }  B''' \in \mathcal{R}_{[\![\lambda']\!]} (B') $,
we have:
\begin{align*}
& (B'''',  \ucontext) \models   \mathit{tr}(\lambda'. \pi)  \text{ for all }  B'''' \in \mathcal{R}_{[\![\lambda']\!]} (B'').
\end{align*}

The latter is equivalent to
$(B'', \ucontext ) \models [\![\lambda']\!] \mathit{tr}(\lambda'. \pi)$
which  is equivalent to
$(B'', \ucontext ) \models \mathit{tr}( Q p \lambda'. \pi ) $.

In an analogous way,
we can prove that
$(B'', \ucontext ) \models \mathit{tr}( Q p \lambda'. \pi ) $
implies
$(B', \ucontext ) \models \mathit{tr}( Q p \lambda'. \pi ) $.
\end{proof}

We can now go back to the statement of Lemma \ref{intermediatelemma2}
and prove it.
The proof of  the lemma is by induction on the length of the formula
$\lambda$.
\paragraph{Base case}
Suppose $\lambda $
is the empty sequence
and $\states \models  \pi $.
By construction of $B^{(\chi,\states)}$,
the latter is equivalent to
$ B^{(\chi,\states)} \models \pi$.
The latter is equivalent to $\big(  B^{(\chi,\states)},      \ucontext  \big)\models \pi$.

 \paragraph{Inductive case}
 We prove that the statement
 is true for sequence $\lambda$ of length $k+1$, if
 we suppose it is true
for sequence $\lambda$ of length $k$.
Suppose $\lambda$
has length $k+1$.
Therefore, $\chi$
can be written as $ Q p  \lambda'. \pi $.
We assume  $Q = \exists$,
as the proof for $Q = \forall$ is analogous.

($\Rightarrow$)
 Suppose $\states \models \exists p \lambda'. \pi $.
 The latter is equivalent to saying that
$\states \cup \{p \} \models   \lambda'. \pi$ or $ \states \setminus \{p \} \models  \lambda'. \pi$.

By induction hypothesis,
the latter implies that
\begin{align*}
&\big( B^{( \lambda'. \pi,\states \cup \{p \})}, \ucontext \big) \models \mathit{tr}( \lambda'. \pi), \text{or}\\
&\big( B^{( \lambda'. \pi,\states \setminus \{p \})}, \ucontext \big) \models \mathit{tr}( \lambda'. \pi).
\end{align*}
It is easy to check that
\begin{align*}
& B^{( \chi,\states )}  \ \relstate{1}  \ B^{( \lambda'. \pi,\states \cup \{p \})},  \text{ and } \\
 & B^{( \chi,\states )} \ \relstate{1}  \ B^{( \lambda'. \pi,\states \setminus \{p \})}.
\end{align*}
Moreover,
\begin{align*}
& \big(B^{( \lambda'. \pi,\states \cup \{p \})}, \ucontext \big) \models    \uncertain{1}\lambda'[0]    ,  \text{ and } \\
& \big(B^{( \lambda'. \pi,\states \setminus \{p \})}, \ucontext \big) \models    \uncertain{1}\lambda'[0] .
\end{align*}
  Therefore,
  \begin{align*}
  \big(B^{( \chi,\states  )}, \ucontext \big) \models   \langle\!\langle\lambda'\rangle\!\rangle
  \mathit{tr} ( \lambda'. \pi) .
  \end{align*}
  The latter is equivalent to   $\big(B^{( \chi,\states)}, \ucontext \big) \models   \mathit{tr} ( \chi) $.

($\Leftarrow$) Suppose
 $\big(B^{( \chi,\states)}, \ucontext \big) \models   \mathit{tr} ( \exists p \lambda'. \pi) $.
 Hence,
  $\big(B^{( \chi,\states)}, \ucontext \big) \models   \langle\!\langle\lambda'\rangle\!\rangle
  \mathit{tr} ( \lambda'. \pi)$.
  The latter implies that
    \begin{align*}
    \exists B' \in  \mathcal{R}_{[\![\lambda']\!]} \big( B^{( \chi,\states)} \big) \text{ such that }
   \big(B', \ucontext \big) \models \mathit{tr} ( \lambda'. \pi)
      \end{align*}
      where $\mathcal{R}_{[\![\lambda']\!]}$
      is defined as in the proof of Lemma \ref{intermediatelemma1}.

  For all
      $ B' \in  \mathcal{R}_{[\![\lambda']\!]} \big( B^{( \chi,\states)} \big)$,
we clearly have:
\begin{align*}
& (B',  \ucontext) \models  \uncertain{1}\lambda'[0] .
\end{align*}
Moreover,
      for all
      $ B' \in  \mathcal{R}_{[\![\lambda']\!]} \big( B^{( \chi,\states)} \big)$,
      we have:
\begin{align*}
& B_1^{(\lambda'. \pi,\states \cup \{p \})} \subseteq B_1' ;\\
& B_1^{(\lambda'. \pi,\states \setminus \{p \})} \subseteq B_1';\\
&\big( (\states \cup \{p \}) \cap \PROP(\pi) \big) = \big( \states' \cap \PROP(\pi)  \big) \text{ or } \\
&  \big( (\states \setminus \{p \}) \cap \PROP(\pi) \big) = \big( \states'  \cap \PROP(\pi) \big).
\end{align*}
Finally, we have:
\begin{align*}
&  B^{(\lambda'. \pi,\states \cup \{p \})} \in  \mathcal{R}_{[\![\lambda']\!]} \big( B^{( \chi,\states) }\big)
\text{ and }
 B^{(\lambda'. \pi,\states \setminus \{p \})} \in  \mathcal{R}_{[\![\lambda']\!]} \big( B^{( \chi,\states) }\big).
\end{align*}
Thus, by Lemma \ref{intermediatelemma1}
and the fact that there exists
$ B' \in  \mathcal{R}_{[\![\lambda']\!]} \big( B^{( \chi,\states)} \big) $ such that
   $\big(B', \ucontext \big) \models \mathit{tr} ( \lambda'. \pi)$,
we have:
\begin{align*}
& \big(B^{(\lambda'. \pi,\states \cup \{p \})},  \ucontext\big) \models  \mathit{tr} ( \lambda'. \pi) \text{, or } \\
& \big(B^{(\lambda'. \pi,\states \setminus \{p \})},  \ucontext\big) \models \mathit{tr} ( \lambda'. \pi).
\end{align*}
  By induction hypothesis,
  the latter implies that
  $\states \cup \{p \} \models  \lambda'. \pi$ or $\states \setminus \{p \} \models  \lambda'. \pi$.
  The latter is equivalent to $\states \models \exists p  \lambda'. \pi $.
\end{proof}

\end{document}